\documentclass[letterpaper, 10 pt, conference]{ieeeconf}  
\usepackage{algorithm}
\usepackage{algorithmicx}
\usepackage{stfloats}
\usepackage{graphicx}          
\usepackage{mathrsfs}         
\usepackage{amsmath}   
\usepackage{amsfonts}    
\usepackage{amssymb}   
\usepackage{booktabs}
\usepackage{threeparttable}
\usepackage{subfigure}
\usepackage{tabularx}
\usepackage{comment}
\usepackage[dvipsnames]{xcolor}
\usepackage{color}

\usepackage[font=small,skip=0pt]{caption}

\usepackage{enumerate}

\usepackage{amsthm}

\newtheorem{remark}{\textbf{Remark}}
\newtheorem{definition}{Definition}
\newtheorem{assumption}{\emph{\textbf{Assumption}}}
\newtheorem{theorem}{\textbf{Theorem}}

\newtheorem{corollary}{Corollary}

\newtheorem{example}{Example}

\DeclareMathOperator{\rank}{rank}
\DeclareMathOperator{\diag}{diag}

\DeclareMathOperator{\tr}{tr} 
\DeclareMathOperator{\im}{im} 

\IEEEoverridecommandlockouts                   
\title{\LARGE \bf
	Decomposability and Parallel Computation of Multi-Agent LQR
}

\author{Gangshan Jing,~He Bai,~Jemin George~and~Aranya Chakrabortty
	\thanks{G.~Jing and A. Chakrabortty are with  North Carolina State University, Raleigh, NC 27695, USA.
		{\tt\small \{gjing, achakra2\}@ncsu.edu}}%
	\thanks{H.~Bai is with Oklahoma State University, Stillwater, OK 74078, USA.
		{\tt\small he.bai@okstate.edu}}%
	\thanks{J.~George is with the US Army Research Laboratory, Adelphi, MD 20783, USA.
		{\tt\small jemin.george.civ@mail.mil}}%
}

\begin{document}

	\maketitle
	\thispagestyle{empty}
	\pagestyle{empty}

	\begin{abstract}
		Individual agents in a multi-agent system (MAS) may have decoupled open-loop dynamics, but a cooperative control objective usually results in coupled closed-loop dynamics thereby making the control design computationally expensive. The computation time becomes even higher when a learning strategy such as reinforcement learning (RL) needs to be applied to deal with the situation when the agents dynamics are not known. To resolve this problem, we propose a parallel RL scheme for a linear quadratic regulator (LQR) design in a continuous-time linear MAS. The idea is to exploit the structural properties of two graphs embedded in the $Q$ and $R$ weighting matrices in the LQR objective to define an orthogonal transformation that can convert the original LQR design to multiple decoupled smaller-sized LQR designs. We show that if the MAS is homogeneous then this decomposition retains closed-loop optimality. Conditions for decomposability, an algorithm for constructing the transformation matrix, a parallel RL algorithm, and robustness analysis when the design is applied to non-homogeneous MAS are presented. Simulations show that the proposed approach can guarantee significant speed-up in learning without any loss in the cumulative value of the LQR cost.
	\end{abstract}


	\section{Introduction}
	Optimal control of multi-agent systems (MASs) has a long-standing literature in both model-based \cite{Borrelli08}-\cite{Xue16}, and model-free \cite{Vrabie09}-\cite{Jing20} settings. One common challenge in both problems, however, is the high computational cost of the control design that often stems from the large size of typical MASs. Individual agents in the MAS may have decoupled open-loop dynamics, but the control objective is usually cooperative in nature which results in coupled closed-loop dynamics thereby making the optimal control design large-sized. The problem becomes even more critical when the controller needs to be {\it learned} in real-time using, for example, learning strategies such as reinforcement learning (RL) during situations when the agents dynamics are not known \cite{Jiang12,Lewis12}.
	
	In this paper, we propose a parallel computation scheme for infinite-horizon linear quadratic regulator (LQR) optimal control of continuous-time homogeneous MAS to resolve the issue of high computation time. The fundamental idea is to define the $Q$ and $R$ matrices of the LQR objective function over a set of communication graphs $G_1$ and $G_2$, respectively, and find a transformation matrix based on the structure of these two graphs. By utilizing this matrix, the original large-size LQR problem is equivalently converted to multiple decoupled smaller-size LQR problems. The transformation matrix itself is structured in the sense that it partitions the agents into a discrete set of non-overlapping groups, where each group solves a decoupled small-size LQR problem. Due to the decoupling, all of these designs can be run independently and in parallel, thereby saving significant amounts of computational effort and time. The design holds for both model-based and model-free LQR. For the sake of this paper, we only focus on the model-free case, and develop a RL learning strategy that is compatible with the reduced-dimensional LQR designs. 
	
	Dimensionality reduction has been used in the past on many occasions to improve the computational efficiency of optimal control, in both model-based \cite{Nguyen16,Xue16} and model-free \cite{Mukherjee18}-\cite{TCNS} settings, but the main difference between these approaches and our proposed approach is that the former designs all result in sub-optimal controllers while our controller, when the MAS is homogeneous, retains the optimality of the original LQR problem. We identify specific conditions on the cost function to establish this optimality. Another important difference is that in conventional designs dimensionality reduction usually happens due to grouping, time-scale separation, or spatial-scale separation in the agent dynamics. In our design, however, the reduction occurs due to the structure imposed on the control objective, not due to the plant dynamics. 
	
	Our results are presented in the following way. We first establish the notion of decomposability for optimal control of homogeneous MAS using Definition \ref{define decompose}, followed by the derivation of multiple sufficient conditions for decomposability in Theorems \ref{th SD} to \ref{th commute}. An algorithm (Algorithm \ref{alg:1}) for constructing the transformation matrix is proposed, and a corresponding hierarchical RL algorithm (Algorithm \ref{alg:2}) is developed to parallelize the control design in situations where the agent models may not be known. Finally, a detailed robustness analysis of the parallel controller is presented to encompass stability, performance, and implementation challenges when this controller is applied to a non-homogeneous MAS. 
	
	The rest of the paper is organized as follows.  Section \ref{sec: problem} introduces the main problem formulation and defines decomposability. Section \ref{sec: conditions} presents multiple conditions for decomposability and proposes Algorithm \ref{alg:1} for construction of the transformation matrix. Section \ref{sec: RL} develops the parallel RL algorithm. Section \ref{sec: robustness} presents robustness analysis. Section \ref{sec: simulation} shows a simulation example. Section \ref{sec: conclusion} concludes the paper.

	\textbf{Notation}: Throughout the paper, $\mathcal{G}=(\mathcal{V},\mathcal{E})$ denotes an unweighted undirected graph with $N$ vertices, where $\mathcal{V}=\{1,...,N\}$ is the set of vertices, $\mathcal{E}\subset\mathcal{V}\times\mathcal{V}$ is the set of edges; Graph $\mathcal{G}$ is said to be disconnected if there are two nodes $i$ and $j$ with no path between them, i.e., there does not exist a sequence of distinct edges of the form $(i_1,i_2)$, $(i_2,i_3)$, ..., $(i_{r-1},i_r)$ where $i_1=i$ and $i_r=j$. The $d\times d$ identity matrix is denoted by $I_d$.  The Kronecker product is denoted by $\otimes$. Given a matrix $X$, $X\succeq0$ implies that $X$ is positive semi-definite;  $\im(X)$ denotes the image space of $X$. We use $\diag\{A_1,...,A_N\}$ to denote a block diagonal matrix with $A_i$'s on the diagonal. Matrix $G_{ab}(s)$ denotes the transfer function from input $a$ to output $b$. Given a transfer function $G(s)$, its $\mathcal{H}_\infty$ norm is $||G(s)||_\infty=\sup_\omega\sigma_{\max}(G(j\omega))$, its $\mathcal{H}_2$ norm is $||G(s)||_2=(\frac{1}{2\pi}\int_{-\infty}^{\infty}\tr(G^\top(j\omega)G(j\omega))d\omega)^{1/2}$. The Euclidean norm is denoted by $||\cdot||$.

	\section{Problem Statement}\label{sec: problem}
	Consider the following linear homogeneous MAS:
	\begin{equation}\label{MAS}
	\dot{x}_i=Ax_i+Bu_i, ~~~~i=1,..., N
	\end{equation}
	where $x_i\in\mathbb{R}^n$ and $u_i\in\mathbb{R}^m$ denote the state and the control input of agent $i$, respectively. Throughout this paper, we assume that $A$ and $B$ are unknown, but the dimension of $x$ and $u$ are known. Let $x=(x_1^{\top},...,x_N^{\top})^{\top}\in\mathbb{R}^{nN}$ and $u=(u_1^{\top},...,u_N^{\top})\in\mathbb{R}^{mN}$, the optimal control problem to be solved in our paper is formulated as
	\begin{equation}\label{original}
	\begin{split}
	\min_u&~~J(x(0),u)=\int_0^\infty (x^{\top}Qx+u^{\top}Ru)dt\\
	\text{s.t.} &~~ \dot{x}=(I_N\otimes A)x+(I_N\otimes B)u.
	\end{split}
	\end{equation}
	where,
	\begin{equation}\label{QR hom}
	Q=G_1\otimes Q_0,~~~~R=G_2\otimes R_0.
	\end{equation}
	The two graphs $G_1\succeq0\in\mathbb{R}^{N\times N}$ and $G_2\succ0\in\mathbb{R}^{N\times N}$  characterize the couplings between the different agents in their desired transient cooperative behavior, with $Q_0\succ0\in\mathbb{R}^{n\times n}$ and  $R_0\succ0\in\mathbb{R}^{m\times m}$. 
	
	RL algorithms for solving LQR control in the absence of $A$ and $B$ have been introduced in \cite{Jiang12,Lewis12}. However, naively applying these algorithms to a large-size network of agents would involve repeated inversions of large matrices, making the overall design computational  expensive. Depending on the values of $n$ and $N$, the learning time in that case can become unacceptably high. To resolve this problem, in the following we will study the situations when the LQR control problem can be equivalently decomposed into multiple smaller-size decoupled LQR problems. For this, we define the notion of {\it decomposability} of  (\ref{original}) as follows.
	
	\begin{definition}\label{define decompose}
		Problem (\ref{original}) is said to be decomposable if there exist $r>1$ functions $J_i(\xi_i(0),v_i)$ such that 
		\begin{equation}\label{decomposeJ}
		J(x(0),u)=\sum_{i=1}^rJ_i(\xi_i(0),v_i),
		\end{equation}
		\begin{equation}\label{block J}
		J_i(\xi_i(0),v_i)=\int_0^\infty(\xi_i^{\top}Q_i\xi_i+v_i^{\top}R_iv_i)dt,
		\end{equation}
		and
		\begin{equation}\label{xi dynamics}
		\dot{\xi}_i=(I_{N_i}\otimes A)\xi_i+(I_{N_i}\otimes B)v_i, ~~i=1,...,r, 
		\end{equation}
		where $Q_i\in\mathbb{R}^{nN_i\times nN_i}$, $R_i\in\mathbb{R}^{mN_i\times mN_i}$, $\xi_i\in\mathbb{R}^{nN_i}$ and $v_i\in\mathbb{R}^{mN_i}$, $\sum_{i=1}^rN_i=N$.
	\end{definition}
	
	\begin{definition}
		Problem (\ref{original}) is said to be completely decomposable if it is decomposable with $r=N$.
	\end{definition}
	
	Let $\xi=(\xi_1^{\top},..., \xi_r^{\top})\in\mathbb{R}^{Nn}$. It is observed from Definition \ref{define decompose} that the dynamics of $\xi$ is identical to $x$. In fact, $\xi_i$ is a vector stacking up states of $N_i$ agents, and therefore can be viewed as the state vector of a cluster containing partial agents in the whole group. When problem (\ref{original}) is completely decomposable, each agent is viewed as a cluster.
	
	\section{Recognition of Decomposable Optimal Control Problems}\label{sec: conditions}
	
	In this section, we propose several conditions for decomposability of the problem (\ref{original}), followed by an algorithm to construct a transformation matrix for decomposition. The algorithm can also be used to identify if a given problem is decomposable.
	
	\subsection{Conditions for Decomposability}
	We start by defining {\it simultaneously block-diagonalizability} of two matrices. 
	
	\begin{definition}\label{de sbd}
		Two matrices $X\in\mathbb{R}^{N\times N}$ and $Y\in\mathbb{R}^{N\times N}$ are  simultaneously block-diagonalizable with respect to a disconnected graph $\mathcal{G}$ if there exists an orthogonal matrix $T$ such that $TXT^{-1}$ and $TYT^{-1}$ are both block-diagonal, and $TXT^{-1},TYT^{-1}\in \mathcal{S}(\mathcal{G})$ for graph $\mathcal{G}=(\mathcal{V},\mathcal{E}_\mathcal{G})$. Here $$\mathcal{S}(\mathcal{G})\triangleq\{M\in\mathbb{R}^{N\times N}: M_{ij}=0 \text{ if } (i,j)\notin\mathcal{E}_{\mathcal{G}}\}$$
		denotes the set of matrices with sparsity patterns similar to that of the graph $\mathcal{G}$.
	\end{definition}

	Since a disconnected graph always has multiple independent connected components, each block on the diagonal of $TXT^{-1}$ and $TYT^{-1}$ corresponds to one or multiple independent connected components in $\mathcal{G}$. Moreover, due to the property of block-diagonal matrices, graph $\mathcal{G}$ in Definition \ref{de sbd} must be undirected and may have self-loops. The following theorem presents a condition for decomposability of (\ref{original}).
	
	\begin{theorem}\label{th SD}
		Problem (\ref{original}) is decomposable if $G_1$ and $G_2$ are simultaneously block-diagonalizable with respect to some graph $\mathcal{G}$.
	\end{theorem}
	
	\begin{proof}
		Since $G_1$ and $G_2$ are simultaneously block-diagonalizable, there exists an orthogonal matrix $T\in\mathbb{R}^{N\times N}$ such that 
		\begin{equation}\label{TG1T}
		TG_1T^{-1}=\Phi=\diag\{\Phi_1,...,\Phi_r\},
		\end{equation}
		\begin{equation}\label{TG2T}
		TG_2T^{-1}=\Psi=\diag\{\Psi_1,...,\Psi_r\},
		\end{equation}
		where $\Phi_i\in\mathbb{R}^{N_i\times N_i}$, $\Psi_i\in\mathbb{R}^{N_i\times N_i}$. We will show that the homogeneous MAS (\ref{MAS}) can be clustered into $r$ subgroups, with the $i$-th subgroup including $N_i$ agents.
		
		Define $A_i=I_{N_i}\otimes A$, $B_i=I_{N_i}\otimes B $, $\bar{T}=T\otimes I_n$,
		\begin{equation}\label{de xi}
		\xi=(\xi_1^{\top},..., \xi_r^{\top})^{\top}=(T\otimes I_n)x\in\mathbb{R}^{nN}, 
		\end{equation}
		\begin{equation}\label{de v}
		v=(v_1^{\top},..., v_r^{\top})^{\top}=(T\otimes I_m)u\in\mathbb{R}^{mN},
		\end{equation}
		where $\xi_i\in\mathbb{R}^{nN_i}$, $v_i\in\mathbb{R}^{mN_i}$. Then the following holds:
		\begin{equation}
		\begin{split}
		x^{\top}Qx&=(\bar{T}x)^{\top}(T\otimes I_N)(G_1\otimes Q_0)(T\otimes I_N)^{-1}\bar{T}x\\
		&=\xi^{\top}(TG_1T^{-1}\otimes Q_0)\xi\\
		&=\sum_{i=1}^r\xi_i^{\top}(\Phi_i\otimes Q_0)\xi_i.
		\end{split}
		\end{equation}
		Similarly, we have 
		\begin{equation}
		u^{\top}Ru=\sum_{i=1}^rv_i^{\top}(\Psi_i\otimes R_0)v_i.
		\end{equation}
		Furthermore, 
		\begin{equation}
		\begin{split}
		\dot{\xi}&=\bar{T}((I_N\otimes A)x+(I_N\otimes B)u)\\
		&=\bar T\left((I_N\otimes A)\bar{T}^{\top}\xi+(I_N\otimes B)(T^{\top}\otimes I_m)v\right)\\
		&=(I_N\otimes A)\xi+(I_N\otimes B)v.
		\end{split}
		\end{equation}
		This completes the proof.
	\end{proof}
	
	Given a matrix $\Gamma\in\mathbb{R}^{N\times s}$ and matrix $G\in\mathbb{R}^{N\times N}$, $\Gamma$ is said to be $G$-invariant if there exists a matrix $C\in\mathbb{R}^{s\times s}$ such that $G\Gamma=\Gamma C$. Then we have the following result.
	
	\begin{theorem}\label{th G1G2invariant}
		Matrices $G_1$ and $G_2$ are simultaneously block-diagonalizable with respect to a disconnected graph if and only if there exists a matrix $\Gamma\in\mathbb{R}^{N\times s}$ with $\Gamma^{\top}\Gamma=I_s$ and $s<N$ such that $\Gamma$ is both $G_1$-invariant and $G_2$-invariant.
	\end{theorem}
	
	\begin{proof}
		We first prove ``necessity". Let $T\in\mathbb{R}^{N\times N}$ be the orthogonal matrix such that (\ref{TG1T}) and (\ref{TG2T}) hold. We partition $T$ into two parts: $T=[\hat{T},\bar T]^{\top}$, where $\hat{T}\in\mathbb{R}^{N\times N_1}$ and $\bar{T}\in\mathbb{R}^{N\times (N-N_1)}$. Then $\hat{T}^{\top}\hat{T}=I_{N_1}$ and $\bar{T}^{\top}\bar T=I_{N-N_1}$. From (\ref{TG1T}), it holds that 
		$$\hat{T}^{\top}G_1\hat{T}=\Phi_1,~~ \hat{T}^{\top}G_1\bar T=0.$$
		Note that $\hat{T}^{\top}\bar T=0$ and $[\hat{T},\bar T]=T^{\top}$ is non-singular. Then $\im(\bar T)$ is the orthogonal complement space of $\im(\hat T)$. It follows that $$G_1\im(\bar T)\subset\im(\bar T),$$
		implying that $\bar T$ is $G_1$-invariant. Similarly, it can be proved that $\bar T$ is $G_2$-invariant. Therefore, choosing $\Gamma=\bar{T}$ proves the necessity.
		
		Next we prove ``sufficiency". $\Gamma^{\top}\Gamma=I_s$ implies that any pair of columns of $\Gamma$ are orthogonal to each other. Then there must exist another matrix $\Xi\in\mathbb{R}^{N\times N-s}$ such that $[\Gamma,\Xi]$ is orthogonal. Let $T=[\Gamma,\Xi]^{\top}$, we have
		$$TG_iT^{-1}=\begin{pmatrix}
		\Gamma^{\top}G_i\Gamma&\Gamma^{\top}G_i\Xi\\
		\Xi^{\top}G_i\Gamma&\Xi^{\top}G_i\Xi
		\end{pmatrix}, ~i=1,2.$$
		Since $\Gamma$ is $G_1$-invariant and $G_2$-invariant, $\Gamma^{\top}G_i\Xi=0$ for $i=1,2$. Therefore, $G_1$ and $G_2$ are simultaneously block-diagonalizable.
	\end{proof}
	
	Specifically, when $s=1$, we obtain the following corollary immediately.
	\begin{corollary}
		Matrices $G_1$ and $G_2$ are simultaneously block-diagonalizable if they have a common eigenvector.
	\end{corollary}
	
	The following theorem gives another equivalent condition for $G_1$ and $G_2$ to be simultaneously block-diagonalizable.
	\begin{theorem}\label{th E1E2}
		Matrices $G_1$ and $G_2$ are simultaneously block-diagonalizable with respect to graph $\mathcal{G}$ if and only if there exist orthogonal matrices $E_1$ and $E_2$ such that $E_1G_1E_1^{\top}$ and $E_2G_2E_2^{\top}$ are diagonal, and $E_1E_2^{\top}\in \mathcal{S}(\mathcal{G})$.
	\end{theorem}
	\begin{proof}
		We first prove ``necessity". Since $G_1$ and $G_2$ are both symmetric, they are both diagonalizable. It follows that there exist orthogonal matrices $U\in\mathbb{R}^{N\times N}$ and $V\in\mathbb{R}^{N\times N}$ such that
		$$TG_1T^{-1}=U\diag\{\lambda_1,...,\lambda_N\}U^{\top}=\diag\{\Phi_1,...,\Phi_r\},$$
		$$TG_2T^{-1}=V\diag\{\mu_1,...,\mu_N\}V^{\top}=\diag\{\Psi_1,...,\Psi_r\}.$$
		Let $\mathcal{G}$ be a disconnected graph such that $TG_1T^{-1}\in \mathcal{S}(\mathcal{G})$ and $TG_2T^{-1}\in \mathcal{S}(\mathcal{G})$. Then there must exist feasible $U$ and $V$ such that $$U,V\in \mathcal{S}(\mathcal{G}),$$
		because such a matrix $U$ can be obtained by designing each block in $U^{\top}$ as a collection of eigenvectors of a block in $\Phi$, so does $V$. Let $E_1=U^{\top}T$ and $E_2=V^{\top}T$, then each column of $E_i^{\top}$ is an eigenvector of $G_i$, $i=1,2$. It follows that 
		\begin{equation}\label{UE1=VE2}
		UE_1=VE_2.
		\end{equation}
		Due to the fact $U,V\in \mathcal{S}(\mathcal{G})$, we have
		\begin{equation}\label{E2E1}
		E_2E_1^{\top}=U^{\top}V\in \mathcal{S}(\mathcal{G}).
		\end{equation}
		This implies that $E_1E_2^{\top}\in \mathcal{S}(\mathcal{G})$.
		
		Next we prove ``sufficiency". Suppose $E_1$ and $E_2$ exist and (\ref{E2E1}) holds.  By setting $U=E_1E_2^{\top}\in \mathcal{S}(\mathcal{G})$, $V=(E_1E_2^{\top})^2$ and $T=UE_1=VE_2=E_1E_2^{\top}E_1$, we obtain that 
		$$TG_1T^{-1}=UE_1G_1E_1^{\top}U^{\top}\in \mathcal{S}(\mathcal{G}),$$
		$$TG_2T^{-1}=VE_2G_2E_2^{\top}U^{\top}\in \mathcal{S}(\mathcal{G}).$$
		This completes the proof.
	\end{proof}

	Also note that two diagonalizable matrices are simultaneously diagonalizable if and only if they commute \cite[Theorem 1.3.12]{Horn12}. As the cost function is considered to be in a quadratic form, both $G_1$ and $G_2$ are symmetric and thus diagonalizable. Hence, we have the following theorem.
	\begin{theorem}\label{th commute}
		Problem (\ref{original}) is completely decomposable if $G_1$ and $G_2$ commute.
	\end{theorem}
	
	\begin{remark}
		In practice, $G_2$ is usually a diagonal matrix when there are no relative input efforts between different agents to be minimized. In such scenarios, the diagonal elements of $G_2$ can be designed according to the eigenvectors of $G_1$ such that the condition in Theorem \ref{th E1E2} is satisfied. Specifically, when $G_2=I_N$, which happens in many scenarios, the LQR design problem (\ref{original}) is always completely decomposable because $I_N$ commutes with any $N$-dimensional matrix. In this case, a suitable choice of $T$ is the matrix such that each column of $T^{\top}$ is an eigenvector of $G_1$.
	\end{remark}

	\subsection{Construction of the Transformation Matrix}
	
	In this subsection, we study how to check if $G_1$ and $G_2$ are simultaneously block-diagonalizable and construct the transformation matrix $T$ (if it exits) such that $TG_1T^{-1}$ and $TG_2T^{-1}$ are simultaneously diagonalized. According to the proof of Theorem \ref{th E1E2}, to find an appropriate transformation matrix $T$, we only need to find appropriate $E_1$ and $E_2$ such that $E_1E_2^{\top}\in \mathcal{S}(\mathcal{G})$ for some $\mathcal{G}$.
	
	For $l=1, 2$, let $E_l$ be an orthogonal matrix such that $E_l^{\top}$ collects eigenvectors of $G_l$. The sparsity pattern of $E_1E_2^{\top}$ is always not unique due to the following two facts: (i) when $G_l$ has an eigenvalue with its multiplicity\footnote{Since both $G_1$ and $G_2$ are symmetric, for each eigenvalue of $G_1$ or $G_2$, its geometric multiplicity is equal to its algebraic multiplicity. Hence, we use the term ``multiplicity" to denote either of them.} more than 1, the corresponding orthonormal eigenvectors lie in a space containing infinite number of vectors; (ii) even when both $G_1$ and $G_2$ have $N$ distinct eigenvalues, different sequences of those eigenvectors in $E_1^{\top}$ and $E_2^{\top}$ lead to different sparsity patterns of $E_1E_2^{\top}$. 
	
	The problem of finding appropriate $E_1$ and $E_2$ satisfying the condition in Theorem \ref{th E1E2} can be formulated as the following feasibility problem:
	\begin{equation}\label{find E_1 E_2}
	\begin{split}
	&\text{Find}~~~ E_1,E_2\\
	\text{s.t.}~~& E_1E_2^{\top}\in \mathcal{S}(\mathcal{G}),\\
	&E_lG_lE_l^{\top} \text{ is diagonal}, l=1,2,\\
	&\mathcal{G} \text{ is disconnected.}
	\end{split}
	\end{equation}
	Problem (\ref{find E_1 E_2}) is nonlinear since both $E_1$ and $E_2$ are variables. It becomes more complicated when $G_1$ or $G_2$ has some eigenvalues with multiplicity more than 1. Moreover, the solution to (\ref{find E_1 E_2}) may not correspond to the graph $\mathcal{G}$ with the largest number of connected components. In what follows, we deal with the special case where the following assumption holds:
	\begin{assumption}\label{as distinct}
		Both $G_1$ and $G_2$ have $N$ distinct eigenvalues.
	\end{assumption}

	Under Assumption \ref{as distinct}, the eigenspace for each eigenvalue has dimension 1. Therefore, we only need to figure out how to order the columns of $E_1$ and $E_2$ such that $E_1E_2^{\top}$ is block-diagonalizable. We present Algorithm \ref{alg:1} for constructing the transformation matrix $T$. It is worth noting that under Assumption \ref{as distinct}, the matrix $T$ obtained by Algorithm \ref{alg:1} leads to the graph $\mathcal{G}$ with the largest number of connected components.
	
	\begin{algorithm}[h]
		\caption{Constructing the Transformation Matrix $T$}\label{alg:1}
		\textbf{Input}: $G_1$ and $G_2$ satisfying Assumption \ref{as distinct}.\\
		\textbf{Output}: $T$.
		\begin{itemize}
			\item[1.] Perform eigenvalue decomposition of $G_1$ and $G_2$, respectively. Obtain $F_1=[p_1,...,p_N]^{\top}$ and $F_2=[q_1,...,q_N]^{\top}$ such that the $i$-th column of $F_l^{\top}$ is an eigenvector of matrix $G_l$, $l=1,2$. 
			\item[2.] Find two partitions  $\Pi_l=\{\mathcal{V}_1^l,...,\mathcal{V}_r^l\}$ for $l=1,2$ with  $\cap_{k=1}^r\mathcal{V}_k^l=\mathcal{V}$ such that $|\mathcal{V}_k^1|=|\mathcal{V}_k^2|=N_k\geq1$ for $k=1,...,r$, and for any $i\in\mathcal{V}_{k_1}^1$ and $j\in\mathcal{V}_{k_2}^2$ with $k_1\neq k_2$, it holds that $p_{i}^{\top}q_{j}=0$. If such a partition cannot be found for $r>1$, then $G_1$ and $G_2$ are not simultaneously diagonalizable; otherwise go to Step 3.
			
			\item[3.] According to the partitions $\Pi_1$ and $\Pi_2$, obtain two index sequences $\{\pi_l(1),...,\pi_l(N)\}$, $i=1,2$ such that $$\mathcal{V}_k^l=\{\pi_l(\sum_{\tau=0}^{k-1}N_\tau+1),...,\pi_l(\sum_{\tau=0}^{k}N_\tau)\}, k=1,...,r,$$ where $N_0=0$. Then construct two permutation matrices $P_l\in\mathbb{R}^{N\times N}$, $l=1,2$ such that 
			\begin{equation}
			P_l(i,j)=\left\{
			\begin{aligned}
			&1,~~\text{if}~ \pi_l(i)=j,\\
			&0,~~\text{otherwise}.
			\end{aligned}
			\right.
			\end{equation}
			Let $E_l=P_lF_l$ for $l=1,2$. Choose $U=E_1E_2^{\top}$.
			\item[4.] Matrix $T$ is obtained by $$T=UE_1=E_1E_2^{\top}E_1.$$
		\end{itemize}
	\end{algorithm}	
	
	\begin{remark}
		Step 2 of Algorithm \ref{alg:1} can be achieved by checking the product of every pair of vectors $p_i$ and $q_j$ for $i,j=1,...,N$. To construct $\Pi_1$ and $\Pi_2$, categorize $i_1$ and $i_2$ into the same cluster $\mathcal{V}_k^1$, and categorize $j$ into cluster $\mathcal{V}_k^2$ if there exists an index $j\in\mathcal{V}$ such that $p_{i_1}^{\top}q_j\neq0$ and $p_{i_2}^{\top}q_j\neq0$. According to the number of elementary operations, the time complexity of Step 2 is $\mathcal{O}(N^2)$. 
	\end{remark}
	
	\begin{remark}
		Note that under Assumption \ref{as distinct}, when $G_1$ and $G_2$ are simultaneously diagonalizable, there are infinite possible choices for $T$. More specifically, to construct $T$, matrix $U$ can be any orthogonal matrix in $\mathcal{S}(\mathcal{G})$. According to different choices of $U$, one can obtain different $V$ and $T$. For example, in Algorithm \ref{alg:1}, based on the choice of $U$, matrix $V$ can be derived from (\ref{UE1=VE2}) as follows: $$V=UE_1(E_2)^{\top}=E_1E_2^{\top}E_1E_2^{\top}.$$ Another feasible choice for $U$, $V$ and $T$ will be
		$$U=E_2E_1^{\top}E_2E_1^{\top},~V=E_2E_1^{\top},~ T=E_2E_1^{\top}E_2.$$
		It can be observed that in both the above two cases, once $E_2E_1^{\top}\in \mathcal{S}(\mathcal{G})$, it always holds that $U,V\in \mathcal{S}(\mathcal{G})$.
	\end{remark}

	\section{Parallel RL Algorithm Design}\label{sec: RL}
	
	In this section, we propose conditions for existence and uniqueness of the solution to each smaller-size problem, and establish the relationship between the lower-dimensional optimal controllers and the optimal controller of the original problem (\ref{original}). Based on this relationship, we propose a parallel RL algorithm to solve (\ref{original}). 
	
	\subsection{Existence and Uniqueness of the Optimal Controller}
	
	Let $\mathcal{A}=I_N\otimes A$, and $\mathcal{B}=I_N\otimes B$. To guarantee the existence and uniqueness of the solution to problem (\ref{original}), we make the following assumption.
	
	\begin{assumption}\label{as control and observe}
		The pair $(\mathcal{A},\mathcal{B})$ is controllable and $(Q^{1/2},\mathcal{A})$ is observable. 
	\end{assumption}
	Assumption \ref{as control and observe} implicitly implies that $G_1$ is non-singular, as proved in the following theorem.
	\begin{theorem}\label{th G1>0}
		Given problem (\ref{original}) with $Q$ and $R$ defined in (\ref{QR hom}), then $(Q^{1/2},\mathcal{A})$ is observable if and only if $G_1\succ0$.
	\end{theorem}
	\begin{proof}
		The ``sufficiency" holds because $Q=G_1\otimes Q_0\succ0$. Next we prove ``necessity". Since $G_1$ is symmetric, there exists an orthogonal matrix $S\in\mathbb{R}^{N\times N}$ such that $$SG_1S^{-1}=\diag\{\lambda_1,...,\lambda_N\}.$$ It suffices to show $\lambda_i>0$ for all $i=1,...,N$. Let $O$ be the observability matrix corresponding to $(Q^{1/2},I_N\otimes A)$. Then $\rank(O)=nN$, and
		$$
		O=\begin{pmatrix}
		G_1^{1/2}\otimes Q_0^{1/2}\\
		G_1^{1/2}\otimes (Q_0^{1/2}A)\\
		\colon\\
		G_1^{1/2}\otimes (Q_0^{1/2}A^{nN-1})
		\end{pmatrix}.
		$$
		Note that $$(SG_1^{1/2}S^{-1})^2=SG_1^{1/2}S^{-1}SG_1^{1/2}S^{-1}=SG_1S^{-1},$$
		implying that $SG_1^{1/2}S^{-1}=(SG_1S^{-1})^{1/2}.$
		As a result,
		\begin{equation}\label{rank(O)}
		\begin{split}
		\bar{O}&\triangleq(I_{nN}\otimes (S\otimes I_n))O( I_{nN}\otimes(S\otimes I_n))^{-1}\\
		&=\begin{pmatrix}
		(S\otimes I_n)(G_1^{1/2}\otimes Q_0^{1/2})(S^{-1}\otimes I_n)\\
		(S\otimes I_n)(G_1^{1/2}\otimes (Q_0^{1/2}A))(S^{-1}\otimes I_n)\\
		\colon\\
		(S\otimes I_n)(G_1^{1/2}\otimes (Q_0^{1/2}A^{nN-1}))(S^{-1}\otimes I_n)
		\end{pmatrix}\\
		&=\begin{pmatrix}
		\diag\{\sqrt{\lambda_1},...,\sqrt{\lambda_N}\}\otimes Q_0^{1/2}\\
		\diag\{\sqrt{\lambda_1},...,\sqrt{\lambda_N}\}\otimes (Q_0^{1/2}A)\\
		\colon\\
		\diag\{\sqrt{\lambda_1},...,\sqrt{\lambda_N}\}\otimes (Q_0^{1/2}A^{nN-1})
		\end{pmatrix}.
		\end{split}
		\end{equation}
		Suppose that $\lambda_i=0$ for some $i\in\{1,...,N\}$, indicating that there are some columns of $\bar{O}$ with all zero elements. Thus, $\rank(O)=\rank(\bar{O})<nN$, which contradicts the observability of $(Q^{1/2},I_N\otimes A)$.
	\end{proof}
	
	When $G_1$ and $G_2$ are simultaneously block-diagonalizable, according to the proof of Theorem \ref{th SD}, it can be equivalently transformed to the following set of decoupled smaller-sized minimization problems: 
	\begin{equation}\label{xii Ji}
	\begin{split}
	\min_{v_i} &J_i(\xi_i,v_i)=\int_0^\infty \left(\xi_i^{\top} (\Phi_i\otimes Q_0)\xi_i+v_i^{\top} (\Psi_i\otimes R_0)v_i\right)dt \\
	&\text{s.t.} \, \dot{\xi}_i= A_i\xi_i+B_iv_i,~~~~ i=1,...,r,    
	\end{split}
	\end{equation}
	where $\xi=(T\otimes I_n)x\in\mathbb{R}^{nN}$, $v=(T\otimes I_m)u\in\mathbb{R}^{mN}$, $A_i=I_{N_i}\otimes A$, $B_i=I_{N_i}\otimes B $.

	\begin{theorem}\label{th block system}
		Consider problem (\ref{original}) with $Q$ and $R$ defined in (\ref{QR hom}). Suppose that $G_1$ and $G_2$ are simultaneously block-diagonalizable. Under Assumption \ref{as control and observe}, the following statements hold:
		
		(i). $(A_i,B_i)$ is controllable for $i=1,...,r$;
		
		(ii). $((\Phi_i\otimes Q_0)^{1/2}, A_i)$ is observable for all $i=1,...,N$.
	\end{theorem}
	\begin{proof}
		(i). Assumption \ref{as control and observe} implies that $(A,B)$ is controllable. According to the definition of controllability, $(A_i,B_i)=(I_{N_i}\otimes A, I_{N_i}\otimes B)$ is controllable for all $i=1,...,r$.
		
		(ii). We first note that $((\Phi\otimes Q_0)^{1/2},\mathcal{A})$ is observable because $\Phi=TGT^{-1}\succ0$. Let $O_i$ be the observability matrix corresponding to $((\Phi_i\otimes Q_0)^{1/2},A_i)$. By  Cayley-Hamilton Theorem, $A^{n}$ can be denoted by a linear combination of $I_n$, $A$, ..., $A^{n-1}$. Therefore,
		\begin{equation}
		\begin{split}
		\rank(O_i)&=\rank
		\begin{pmatrix}
		\Phi_i^{1/2}\otimes Q_0^{1/2}\\
		\Phi_i^{1/2}\otimes Q_0^{1/2}A\\
		\colon\\
		\Phi_i^{1/2}\otimes Q_0^{1/2}A^{nN_i-1}
		\end{pmatrix}\\
		&=\rank\begin{pmatrix}
		\Phi_i^{1/2}\otimes Q_0^{1/2}\\
		\Phi_i^{1/2}\otimes Q_0^{1/2}A\\
		\colon\\
		\Phi_i^{1/2}\otimes Q_0^{1/2}A^{n-1+k}
		\end{pmatrix}
		\end{split}
		\end{equation}
		for any $k\geq0$. Due to observability of $((\Phi\otimes Q_0)^{1/2},\mathcal{A})$, we have $$\rank(O)=\sum_{i=1}^r\rank(O_i)=nN.$$ Together with $\rank(O_i)\leq nN_i$, we have $\rank(O_i)=nN_i$ for all $i=1,...,r$.
	\end{proof}

	\subsection{Parallel RL Algorithm}
	We next propose a Parallel RL algorithm to synthesize the optimal controller when (\ref{original}) is decomposable. Let $u^*=-K^*x$ be the optimal controller for problem (\ref{original}), $v_i^*=-\kappa_i\xi_i$ be the optimal controller for the $i$-th optimal control problem in (\ref{xii Ji}). Define $T_i^{\top}\in\mathbb{R}^{N_i\times N}$ as a matrix consisting of the $N_i$ rows in $T$ corresponding to the $i$-th cluster. That is, $T=[T_1,...,T_r]^{\top}$. From the form of $V$, we have 
	\[
	\begin{split}
	u^*&=\sum_{i=1}^r(T_i\otimes I_m)v_i^*
	=-\sum_{i=1}^r(T_i\otimes I_m)\kappa_i(T_i^{\top} \otimes I_n)x.
	\end{split}
	\]
	It follows that
	\begin{equation}\label{K*}
	\begin{split}
	K^*&=\sum_{i=1}^r(T_i\otimes I_m)\kappa_i(T_i^{\top} \otimes I_n)\\
	&=(T^{\top}\otimes I_m)\kappa(T\otimes I_n),    
	\end{split}
	\end{equation}
	where $\kappa=\diag\{\kappa_1,...,\kappa_r\}$.	To solve for each $\kappa_i$ without knowing $A$ and $B$, one may use \cite[Algorithm 3]{Jing20}. Note that although \cite[Algorithm 3]{Jing20} is designed for solving a minimum variance problem, it also applies to a deterministic LQR problem by setting the coefficients appropriately.

	\begin{algorithm}[h]
		\caption{Parallel RL Algorithm for Optimal Control of Homogeneous Linear MAS}\label{alg:2}
		\textbf{Input}: $G_1$, $G_2$, $T$, $Q_0$, $R_0$.\\
		\textbf{Output}: Optimal control gain $K^*$
		\begin{itemize}
			\item[1.] Obtain scalars $N$, $n$ and $m$ from the dimensions of $G_1$, $Q_0$ and $R_0$. Compute $\Phi_i$ and $\Psi_i$ from (\ref{TG1T}) and (\ref{TG2T}), respectively.		
			\item[2.] Run \cite[Algorithm 3]{Jing20} to solve the $i$-th LQR problem (\ref{xii Ji}) for $i=1,...,r$. Obtain the optimal control gain $\kappa_i$ for the $i$-th LQR problem.  
			\item[3.] Compute the global optimal control gain $K^*$ according to (\ref{K*}).
		\end{itemize}
	\end{algorithm}	
	
	Although Algorithm \ref{alg:1} may need to be implemented in advance to provide inputs for Algorithm \ref{alg:2}, where one needs to execute eigenvalue decomposition for two $N\times N$ matrices once, the overall computational complexity is still much lower than that of computing the inversion of a $(nN(nN+1)/2+mnN^2)$-dimensional matrix, which is required in the conventional RL algorithm \cite{Jiang12}.
	
	\section{Robustness Analysis}\label{sec: robustness}
	In reality, a MAS may not have perfectly homogeneous agents. We next present robustness analysis of our parallel LQR controller when the controller is designed under the assumption of homogeneity as in Algorithm \ref{alg:2}, but is implemented in a heterogenous MAS. 
	
	Suppose the dynamic model of agent $i$ is given as
	\begin{equation}\label{het xi}
	\dot{x}_i=A^h_ix_i+B_i^hu_i, ~~i=1,...,N.
	\end{equation}
	The MAS model is written in a compact form as
	\begin{equation}\label{het x}
	\dot{x}=\mathcal{A}^hx+\mathcal{B}^hu.
	\end{equation}
	Both $\mathcal{A}^h$ and $\mathcal{B}^h$ are unknown. The transfer function of (\ref{het x}) from $u$ to $x$ is
	\begin{equation}
	G(s)=(sI_{nN}-\mathcal{A}^h)^{-1}\mathcal{B}^h.
	\end{equation}
	Throughout this section, we make the following assumption:
	\begin{assumption}\label{as Ah}
		$(\mathcal{A}_i^h,\mathcal{B}_i^h)$ is controllable for $i=1,..., N$, $Q\succ0$.
	\end{assumption}
	Note that $Q\succ0$ is a necessary condition for observability of the optimal control problem for homogeneous MAS, as shown in Theorem \ref{th G1>0}. If $Q$ has at least one zero eigenvalue, the condition for implementation of  Algorithm \ref{alg:2} is not satisfied.

	Let $T\in\mathbb{R}^{N\times N}$ be the transformation matrix such that $TG_1T^{\top}$ and $TG_2T^{\top}$ are simultaneously block-diagonal. Let $K$ be the control gain learned by Algorithm \ref{alg:2}\footnote{Note that Algorithm \ref{alg:2} is always able to converge to some control gain matrix under Assumption \ref{as Ah} because control input and state data are collected on each cluster of individual systems. The model mismatch is only reflected in the control objective of each subproblem (\ref{xii Ji}).}. From the proof of Theorem \ref{th SD}, we know that $K$ is the optimal control gain of the following problem:
	\begin{equation}\label{xhat}
	\begin{split}
	\min_\zeta&~~\hat{J}(\hat{x}(0),\zeta)=\int_0^\infty (\hat{x}^{\top}Q\hat{x}+\zeta^{\top}R\zeta)dt\\
	\text{s.t.} &~~ \dot{\hat{x}}=\bar{T}^{\top}\mathcal{A}^h\bar{T}\hat{x}+\bar{T}^{\top}\mathcal{B}^h\hat{T}\zeta,
	\end{split}
	\end{equation}
	where $\bar{T}=T\otimes I_n$, $\hat{T}=T\otimes I_m$. This implies that $$\hat{\mathcal{A}}\triangleq\bar{T}^{\top}\mathcal{A}^h\bar{T}-\bar{T}^{\top}\mathcal{B}^h\hat{T}K$$ is Hurwitz. Under Assumption \ref{as Ah}, we know $((\bar{T}Q\bar{T}^T)^{1/2},\mathcal{A}^h)$ is observable. Then there is a unique positive definite solution $P$ to the following algebraic Riccati equation:
	\begin{equation}
	\mathcal{A}^hP+P\mathcal{A}^h+\bar{T}Q\bar{T}^T-P\mathcal{B}^h\hat{T}R^{-1}\hat{T}^{\top}\mathcal{B}^hP=0.
	\end{equation}
	Let $\hat{P}=\bar{T}^TP\bar{T}$. Then we have $K=R^{-1}(\bar{T}^{\top}\mathcal{B}^h\hat{T})^T\hat{P}.$
	
	Let $\tilde{\mathcal{A}}=\mathcal{A}^h-\bar{T}^{\top}\mathcal{A}^h\bar{T}$, $\tilde{\mathcal{B}}=\mathcal{B}^h-\bar{T}^{\top}\mathcal{B}^h\hat{T}$. We present stability and performance analysis respectively as follows.
	
\subsection{Stability Analysis}	
	
	The following theorem gives a condition on $\tilde{\mathcal{A}}$ and $\tilde{B}$ such that system (\ref{het x}) with control gain $K$ is stable.
	
	\begin{theorem}\label{thm:lyapunov}
		Consider the optimal control of the heterogeneous MAS (\ref{het x}) with the cost function $J(x(0),u)$ in (\ref{original}). Suppose that $G_1$ and $G_2$ are simultaneously block-diagonalizable. By implementing the optimal control gain learned by Algorithm \ref{alg:2}, MAS (\ref{het x}) is asymptotically stable if
		\begin{multline}\label{LMI}
		\hat{P}\tilde{\mathcal{A}}+\tilde{\mathcal{A}}^T\hat{P}-\hat{P}\tilde{\mathcal{B}}R^{-1}\hat{\mathcal{B}}^T\hat{P}-\hat{P}\hat{\mathcal{B}}R^{-1}\tilde{\mathcal{B}}^T\hat{P}\\-Q-\hat{P}\mathcal{B}^hR^{-1}\mathcal{B}^{hT}\hat{P}\prec0.
		\end{multline}
	\end{theorem}
	
	\begin{proof}
		By applying the controller $u=-Kx$, system (\ref{het x}) becomes
		\begin{equation}
		\dot{x}=(\tilde{\mathcal{A}}-\tilde{\mathcal{B}}K)x.
		\end{equation}
		Consider  $V=x^T\hat{P}x$ as a Lyapunov function. It follows that
		\begin{equation}
		\begin{split}
		\dot{V}&=x^T(\hat{P}\tilde{\mathcal{A}}+\tilde{\mathcal{A}}^T\hat{P})x+x^T(\hat{P}(\tilde{\mathcal{A}}-\tilde{\mathcal{B}}K)+(\tilde{\mathcal{A}}-\tilde{\mathcal{B}}K)^T\hat{P})x\\
		&=x^T(-Q-\hat{P}\mathcal{B}^hR^{-1}\mathcal{B}^{hT}\hat{P})x\\
		&+x^T(\hat{P}\tilde{\mathcal{A}}+\tilde{\mathcal{A}}^T\hat{P}-\hat{P}\tilde{\mathcal{B}}R^{-1}\hat{\mathcal{B}}^T\hat{P}-\hat{P}\hat{\mathcal{B}}R^{-1}\tilde{\mathcal{B}}^T\hat{P})x.
		\end{split}
		\end{equation}
		Condition (\ref{LMI}) implies that $\dot{V}<0$ if $x\neq0$ and $\dot{V}=0$ otherwise. Therefore, MAS (\ref{het x}) with control gain $K$ is asymptotically stable.
	\end{proof}

	Theorem~\ref{thm:lyapunov} indicates stability robustness of our controller by giving a condition on $\tilde A$ and $\tilde B$ associated with $\hat{P}$. When the open loop system~(\ref{het x}) is stable, the following theorem employs the small-gain theorem \cite{Zhou98} to obtain a condition on $\tilde A$ and $\tilde B$ associated with the transfer function $G(s)$ of (\ref{het x}) for the control gain $K$ to be stabilizing.
	\begin{theorem}\label{th G(s) stable}
		Consider the optimal control of the heterogeneous MAS (\ref{het x}) with the cost function $J(x(0),u)$ in (\ref{original}). Suppose that $G_1$ and $G_2$ are simultaneously block-diagonalizable and $\mathcal{A}^h$ is Hurwitz. By implementing the optimal control gain learned by Algorithm \ref{alg:2}, MAS (\ref{het x}) is asymptotically stable if 
		\begin{equation}\label{G(s)condition}
		||G_{\Sigma}(s)||_{\infty}<||K(sI_{nN}-\hat{\mathcal{A}})^{-1}||_{\infty}^{-1},
		\end{equation}
		where 
		\begin{equation}\label{GSigma}
		G_{\Sigma}(s)=(\tilde{A}-\tilde{B}K)G(s).
		\end{equation}
	Specifically, when the MAS is homogeneous, $G_{\Sigma}(s)=0$.
	\end{theorem}
	\begin{proof}
		Applying the controller $u=-Kx$ to (\ref{het x}), then closed-loop system dynamics can be rewritten as the following two interconnected systems:
		\begin{equation}\label{Sigma}
		\Sigma:\left\{
		\begin{aligned}
		&\dot{x}=\mathcal{A}^hx+\mathcal{B}^hu,\\
		&e=(\tilde{A}-\tilde{B}K)x,
		\end{aligned}
		\right.
		\end{equation}	
		and
		\begin{equation}\label{Delta}
		\Delta:\left\{
		\begin{aligned}
		&\dot{\eta}=\hat{\mathcal{A}}\eta+e,\\
		&u=-K\eta.
		\end{aligned}
		\right.
		\end{equation}	
		The feedback interconnection between system $\Sigma$ and system $\Delta$ is interpreted in Fig. \ref{fig SigmaDelta}, where $x(0)$ is the initial state of system (\ref{het x}). Accordingly, $\eta(0)=x(0)$, $u(0)=-K\eta(0)$, $e(0)=(\tilde{A}-\tilde{B}K)x(0)$.	
		\begin{figure}
			\centering
			\includegraphics[width=8cm]{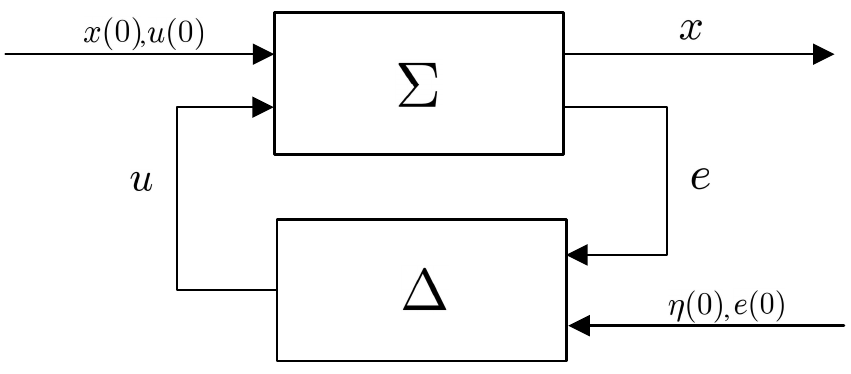}
			\caption{Interconnection between system $\Sigma$ and system $\Delta$.} \label{fig SigmaDelta}
			\vspace{-0.3cm}
		\end{figure}		
		The transfer function of (\ref{Sigma}) from $u$ to $e$ is $G_\Sigma(s)$ in (\ref{GSigma}).	
		On the other hand, the transfer function of (\ref{Delta}) from $e$ to $u$ is
		\begin{equation}
		G_\Delta(s)=-K(sI_{nN}-\hat{\mathcal{A}})^{-1}.
		\end{equation}
		We note that the open loop systems of $\Sigma$ and $\Delta$ are both stable. Moreover, (\ref{G(s)condition}) implies that
		\begin{equation}
		||G_{\Sigma}(s)||_{\infty}||G_{\Delta}(s)||_{\infty}<1.
		\end{equation}
		Using the small-gain theorem, we can conclude that the interconnection between $\Sigma$ and $\Delta$ is stable.
	\end{proof}
	\subsection{Performance Analysis}
	Next we analyze the performance of the heterogenous MAS (\ref{het x}) with controller $u=-Kx$. Let $\sqrt{Q}$ and $\sqrt{R}$ be the two matrices such that $\sqrt{Q}^2=Q$ and $\sqrt{R}^2=R$. Define $y=\begin{pmatrix}
	\sqrt{Q}\\
	-\sqrt{R}K
	\end{pmatrix}x$. It follows that
	\begin{equation}
	J(x(0),u)=\int_0^\infty y^{\top}(t)y(t)dt|_{u=x(0)\delta(t)}=||G_{uy}x(0)||_2^2,
	\end{equation}
	where $\delta(t)=\infty$ if $t=0$, and $\delta(t)=0$ if $t\neq0$, and $\int_{-\infty}^{\infty}\delta(t)dt=1$.
	
	For the convenience of analysis, we evaluate robustness with respect to the following $L_2$-norm directly:
	\begin{equation}\label{J_2}
	J_2(x(0),u)=||y||_2=||G_{uy}x(0)||_2.
	\end{equation}
	Let $\bar{J}_2=J_2(\hat{x}(0),-K\hat{x})$ be the optimal performance for the $\hat{x}$ dynamics in (\ref{xhat}), where $\hat{x}(0)=x(0)$. $\bar{J}$ is finite since $K$ is the optimal control gain of problem (\ref{xhat}). The theorem below analyzes the performance $J_2$ of the heterogeneous MAS (\ref{het x}) with controller $u=-Kx$.
	\begin{theorem}\label{th J2<=barJ2}
		Consider the optimal control of the heterogeneous MAS (\ref{het x}) with the cost function $J_2(x(0),u)$ in (\ref{J_2}). Suppose that $G_1$ and $G_2$ are simultaneously block-diagonalizable. By implementing the optimal control gain learned by Algorithm \ref{alg:2}, it holds that
		\begin{equation}\label{J_2}
		J_2\leq \bar{J}_2+\alpha\epsilon,
		\end{equation}
		where $\epsilon=||G_\Sigma(s)||_2$ and $\alpha=||G_{ey}(I-G_\Sigma G_{eu})^{-1}||_2||G_{\delta u}||_2$.
		Specifically, when the MAS is homogeneous, we have $\epsilon=0$.
	\end{theorem}

	\begin{proof}
		We treat $e=(\tilde{A}-\tilde{B}K)x$ as the disturbance, and consider the following system:
		\begin{equation}
		\begin{split}
		\dot{\eta}&=\hat{\mathcal{A}}\eta+e+x(0)\delta,\\
		u&=-K\eta,\\
		y&=\begin{pmatrix}
		\sqrt{Q}\\
		-\sqrt{R}K
		\end{pmatrix}\eta.
		\end{split}
		\end{equation}
		In the frequency domain, we have
		\begin{equation}
		\begin{pmatrix}
		y\\
		u
		\end{pmatrix}=\begin{pmatrix}
		G_{\delta y}& G_{ey}\\
		G_{\delta u}& G_{eu}
		\end{pmatrix}\begin{pmatrix}
		\delta\\e
		\end{pmatrix}.
		\end{equation}
		It follows that
		\begin{equation}
		\begin{split}
		&J_2=||y||_2\\
		&=||G_{\delta y}+G_{ey}(I-G_\Sigma G_{eu})^{-1}G_\Sigma G_{\delta u}||_2\\
		&\leq ||G_{\delta y}||_2+||G_{ey}(I-G_\Sigma G_{eu})^{-1}G_\Sigma G_{\delta u}||_2,
		\end{split}
		\end{equation}
		where $G_{\Sigma}$ is in form (\ref{GSigma}). We note that 
		\begin{equation}
		\begin{split}
		||G_{\delta y}||_2&=\left|\left|\begin{pmatrix}
		\sqrt{Q}\\
		-\sqrt{R}K
		\end{pmatrix}(sI_{nN}-\hat{\mathcal{A}})^{-1}x(0)\right|\right|_2\\
		&=\bar{J}_2.\\
		\end{split}
		\end{equation}
		Then we obtain (\ref{J_2}).
	\end{proof}

	Theorem \ref{th J2<=barJ2} establishes an upper bound on $J_2$, which is associated with $||G_\Sigma(s)||_2$. The inequality (\ref{J_2}) implies that when $||G_\Sigma(s)||_2$ is small enough, then the performance with respect to controller $u=-Kx$ will be close to $\bar{J}_2$.

	\section{Numerical Examples}\label{sec: simulation}	
	We validate the proposed controller using two examples of MAS, one homogeneous and one heterogeneous.  When Algorithm \ref{alg:2} is applied, the dynamics of the agents are always considered to be unknown.  

	Consider a MAS with $N=100$ agents. Each agent $i$ is a second-order dynamic system with $$A_i=\begin{pmatrix}
	0&I_2\\
	0&-\frac{c_i}{m_i}I_2
	\end{pmatrix}, ~B_i=\begin{pmatrix}
	0\\
	\frac{1}{m_i}I_2
	\end{pmatrix}, ~i=1,...,N.$$ Let $\mathcal{G}$ be a connected undirected graph with 100 nodes randomly generated in MATLAB, as shown in Fig \ref{fig network}. Matrices $Q$ and $R$ are set as $Q=(0.5I_N+L)\otimes I_n$ and $R=I_{mN}$, where $L$ is the Laplacian matrix of graph $\mathcal{G}$. In this case, $G_1=0.5I_N+L$, $G_2=I_N$, $Q_0=I_n$ and $R_0=I_m$. This formulation can describe formation control in \cite{Borrelli08,TCNS}. Since $G_1$ and $G_2$ commute, according to Theorem \ref{th commute}, the problem is completely decomposable. Matrix $T$ is constructed by collecting the orthonormal eigenvectors of $G_1$. 
	
	In both examples, we use $w\in\mathbb{R}^4$ as the common initial state for all the agents, where $w_i$ is randomly chosen from a $[0,1]$, $i=1,...,4$. Therefore, the initial state for the overall MAS is $x(0)=\mathbf{1}_N\otimes w$. Let $t_{opt}$ be the computational time for solving the model-based problem (\ref{original}) via MATLAB, $K_{opt}$ be the obtained optimal control gain, and $J_{opt}$ be the corresponding performance. Let $t_{HRL}$ denote the computational time for implementing Algorithm \ref{alg:2} to solve the model-free problem, $K_{HRL}$ be the obtained control gain, and $J_{HRL}$ be the corresponding performance. Also let $t_{RL}$ be the time of solving the model-free problem via the conventional RL algorithm in \cite{Jiang12}.
	
	\begin{example}
		We first consider a homogeneous MAS where $c_i=m_i=1$ for $i=1,\cdots,N$. The simulation results are shown in Table \ref{tab1}, where ``- -" implies the computational time is more than one hour. This example shows that when the homogeneous MAS has a large size, the computational speed of Algorithm \ref{alg:2} can be even higher than that for the conventional LQR control with known $\mathcal{A}$ and $\mathcal{B}$, and the obtained control gain is almost optimal.
		
		\begin{table}[htbp]	
			\centering
			\fontsize{7}{7}\selectfont
			\begin{threeparttable}
				\caption{Algorithms Comparison: Homogeneous MAS}
				\label{tab1}
				\begin{tabular}{cccccc}
					\toprule			
					    $t_{opt}$(s)&$t_{HRL}$(s)&$t_{RL}$(s)&$J_{opt}$&$J_{HRL}$&$||K_{opt}-K||$\\
					\midrule
					 4.8635 & \textbf{0.1234} & - - & 523.7658 & \textbf{523.7662}&0.0037\\
					\bottomrule
				\end{tabular}
			\end{threeparttable}
		\end{table}
		
		\begin{figure}
			\centering
			\includegraphics[width=8cm]{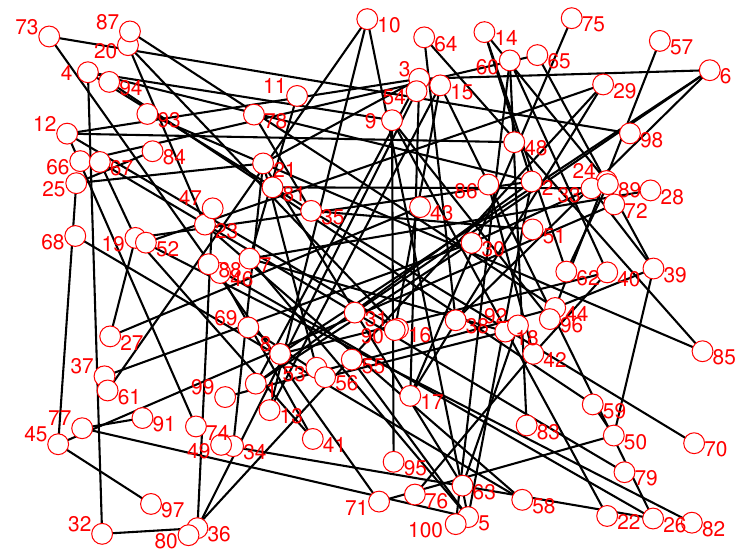}
			\caption{A graph describing the sparsity pattern of $G_1$.} \label{fig network}
		\end{figure}
		\vspace{-0.3cm}
	\end{example}
	
	\begin{example}
		We next consider a heterogeneous MAS, where $c_i=c+\alpha_i$, $m_i=m+\beta_i$, $c=m=1$, $\alpha_i$ and $\beta_i$ are randomly chosen from $[-0.5,0.5]$. The simulation results are shown in Table \ref{tab2}. Note that $\frac{J_{HRL}-J_{opt}}{J_{opt}}=1.94\%$, which implies that Algorithm \ref{alg:2} has a strong robustness on the system performance for this example. 
		\begin{table}[htbp]	
			\centering
			\fontsize{7}{7}\selectfont
			\begin{threeparttable}
				\caption{Algorithms Comparison: Heterogeneous MAS}
				\label{tab2}
				\begin{tabular}{cccccc}
					\toprule			
					    $t_{opt}$(s)&$t_{HRL}$(s)&$t_{RL}$(s)&$J_{opt}$&$J_{HRL}$&$||K_{opt}-K||$\\
					\midrule
					 7.0832 & \textbf{0.1184} & - - & 538.7285 & \textbf{549.1773}&0.8241\\
					\bottomrule
				\end{tabular}
			\end{threeparttable}
		\end{table}
	\end{example}

	
	\section{Conclusions}\label{sec: conclusion}
	We introduced the notion of ``decomposability" for LQR control of homogeneous linear MAS using which a large LQR design problem can be equivalently transformed to multiple smaller-size LQR design problems.  Subsequently, a parallel RL algorithm was proposed to solve these small-size designs and synthesize the optimal controller in a model-free way. Robustness analysis was established, followed by simulation results that clearly show that our parallel RL strategy is much faster than conventional RL.

\end{document}